\documentclass[a4paper]{article}

\usepackage{amsmath}
\usepackage{amssymb}
\usepackage[dvips]{graphicx}
\usepackage[OT4]{fontenc}
\usepackage[latin2]{inputenc}
\usepackage{url}
\usepackage[numbers]{natbib}
\usepackage{bm}
\usepackage{setspace}

\newcommand{\abs}[1]{\left| #1 \right|}

\newcommand{\okra}[1]{\left( #1 \right)}
\newcommand{\kwad}[1]{\left[ #1 \right]}
\newcommand{\klam}[1]{\left\{ #1 \right\}}

\newcommand{\floor}[1]{\left\lfloor #1 \right\rfloor}

\DeclareMathOperator{\sred}{\mathbf{E}}

\newcommand{\boole}[1]{{\bf 1}{\klam{#1}}}

\newtheorem{example}{Example}

\newtheorem{proposition}{Proposition}
\newtheorem{theorem}{Theorem}
\newtheorem{lemma}{Lemma}
\newenvironment*{proof}{\begin{trivlist}\item[]
\noindent\textbf{Proof:}}{$\Box$\par\end{trivlist}}
\newenvironment*{proof*}[1]{\begin{trivlist}\item[]
\noindent\textbf{Proof of #1:}}{$\Box$\par\end{trivlist}}

\title{On Hidden Markov Processes \\ with Infinite Excess Entropy}
\date{}

\author{{\L}ukasz D\k{e}bowski%
  \thanks{
    {\L}. D\k{e}bowski is with the Institute of Computer Science,
    Polish Academy of Sciences, ul. Jana Kazimierza 5, 01-248 Warszawa,
    Poland (e-mail: ldebowsk@ipipan.waw.pl).}  }

\begin{document}

\pagestyle{empty}   
\begin{titlepage}
\maketitle

\begin{abstract}
  We investigate stationary hidden Markov processes for which mutual
  information between the past and the future is infinite. It is
  assumed that the number of observable states is finite and the
  number of hidden states is countably infinite. Under this
  assumption, we show that the block mutual information of a~hidden
  Markov process is upper bounded by a~power law determined by the
  tail index of the hidden state distribution. Moreover, we exhibit
  three examples of processes.  The first example, considered
  previously, is nonergodic and the mutual information between the
  blocks is bounded by the logarithm of the block length. The second
  example is also nonergodic but the mutual information between the
  blocks obeys a~power law. The third example obeys the power law and
  is ergodic.
  \\[1em]
  \\[1em]
  \\[1em]
  \textbf{Running head}: On Processes with Infinite Excess Entropy
\end{abstract}

\end{titlepage}
\pagestyle{plain}   


\section{Introduction}
\label{secIntro}

In recent years there has been a~surge of interdisciplinary interest
in excess entropy, which is the Shannon mutual information between the
past and the future of a~stationary discrete-time process. The initial
motivation for this interest was a~paper by Hilberg \cite{Hilberg90},
who supposed that certain processes with infinite excess entropy may
be useful for modeling texts in natural language. Subsequently, it was
noticed that processes with infinite excess entropy appear also in
research of other, so called, complex systems
\cite{CrutchfieldYoung89,Ebeling97,EbelingPoschel94,Gramss94,
  BialekNemenmanTishby01b,ShaliziCrutchfield01,CrutchfieldFeldman03,Lohr09}.
Also from a~purely mathematical point of view, excess entropy is an
interesting measure of dependence for nominal valued random processes,
where the analysis of autocorrelation does not provide sufficient
insight into process memory. 

Briefly reviewing earlier works, let us mention that excess entropy
has been already studied for several classes of processes. The most
classical results concern Gaussian processes, where Grenander and
Szeg\H{o} \cite[Section 5.5]{GrenanderSzego84} gave an integral
formula for excess entropy (in disguise) and Finch \cite{Finch60}
evaluated this formula for autoregressive moving average (ARMA)
processes. In the ARMA case excess entropy is finite. A few more
papers concern processes over a finite alphabet with infinite excess
entropy. For instance, Bradley \cite{Bradley80} constructed the first
example of a mixing process having this property. Gramss
\cite{Gramss94} investigated a process which is formed by the
frequencies of 0's and 1's in the rabbit sequence. Travers and
Crutchfield \cite{TraversCrutchfield11} researched some hidden Markov
processes with a countably infinite number of hidden states. Some
attempts were also made to generalize excess entropy to
two-dimensional random fields
\cite{FeldmanCrutchfield03,BulatekKaminski09}.

Excess entropy is an intuitive measure of memory stored in
a~stochastic process. Although this quantity only measures the memory
capacity, without characterizing \emph{how} the process future depends
on the past, it can be given interesting general
interpretations. Mahoney, Ellison and Crutchfield
\cite{MahoneyEllisonCrutchfield09,EllisonMahoneyCrutchfield09}
developed a formula for excess entropy in terms of predictive and
retrodictive $\epsilon$-machines, which are minimal unifilar hidden
Markov representations of the process
\cite{ShaliziCrutchfield01,Lohr09}. In our previous works
\cite{Debowski11b,Debowski11d,Debowski12,Debowski10}, we also
investigated excess entropy of stationary processes that model texts
in natural language. We showed that a~power-law growth of mutual
information between adjacent blocks of text arises when the text
describes certain facts in a~logically consistent and highly
repetitive way. Moreover, if the mutual information between blocks
grows according to a power law then a similar power law is obeyed by
the number of distinct words, identified formally as codewords in a
certain text compression \cite{DeMarcken96}.  The latter power law is
known as Herdan's law \cite{Herdan64}, which is an integral version of
the famous Zipf law observed for natural language \cite{Zipf35}.

In this paper we will study several examples of stationary hidden
Markov processes over a~finite alphabet for which excess entropy is
infinite. 
The first study of such processes was developed by Travers and
Crutchfield \cite{TraversCrutchfield11}.  A~few more words about the
adopted setting are in need.
First, excess entropy is finite for hidden Markov chains with a~finite
number of hidden states. This is the usually studied case
\cite{EphraimMerhav02}, for which the name of finite-state sources is
also used.  To allow for hidden Markov processes with unbounded mutual
information, we need to assume that the number of hidden states is at
least countably infinite.
Second, we want to restrict the class of studied models. If we
admitted an uncountable number of hidden states or a~nonstationary
distribution over the hidden states then the class of hidden Markov
processes would cover all processes (over a~countable alphabet). For
that reason we will assume that the underlying Markov process is
stationary and the number of hidden states is exactly countably
infinite. In contrast, the number of observable states is fixed as finite
to focus on nontrivial examples. In all these assumptions we follow
\cite{TraversCrutchfield11}.

The modest aim of the present paper is to demonstrate that power-law
growth of mutual information between adjacent blocks may arise for
very simple hidden Markov processes.  Presumably, stochastic processes
which exhibit this power law appear in modeling of natural language
\cite{Hilberg90,Debowski11b}. But the processes that we study here do
not have a clear linguistic interpretation. They are only mathematical
instances presented to show what is possible in theory. Although these
processes are simple to define, we perceive them as somehow artificial
because of the way \emph{how} the memory of the past is stored in the
present and revealed in the future. Understanding what are acceptable
mechanisms of memory in realistic stochastic models of complex systems
is an important challenge for future research.

The further organization of the paper is as follows: In Section
\ref{secResults} we present the results, whereas the proofs are
deferred to Section \ref{secProofs}.

\section{Results}
\label{secResults}

Now we begin the formal presentation of our results. First, let
$(Y_i)_{i\in\mathbb{Z}}$ be a~stationary Markov process on
$(\Omega,\mathcal{J},P)$ where variables
$Y_i:\Omega\rightarrow\mathbb{Y}$ take values in a~countably infinite
alphabet $\mathbb{Y}$. This process is called the hidden
process. Next, for a~function $f:\mathbb{Y}\rightarrow\mathbb{X}$,
where the alphabet $\mathbb{X}=\klam{0,1,...,D-1}$ is finite, we
construct process $(X_i)_{i\in\mathbb{Z}}$ with
\begin{align}
  X_i=f(Y_i)
  .
\end{align}
Process $(X_i)_{i\in\mathbb{Z}}$ will be called the observable
process. The process is called unifilar if $Y_{i+1}=g(Y_i,X_{i+1})$
for a~certain function
$g:\mathbb{Y}\times\mathbb{X}\rightarrow\mathbb{Y}$. Such
a~construction of hidden Markov processes, historically the oldest one
\cite{Blackwell56}, is called state-emitting (or Moore) in contrast to
another construction named edge-emitting (or Mealy). The Mealy
construction, with a~requirement of unifilarity, has been adopted in
previous works
\cite{CrutchfieldFeldman03,TraversCrutchfield11,Lohr09}. Here, we
adopt the Moore construction and we drop the unifilarity assumption
since it leads to a~simpler presentation of processes. It should be
noted that the standard definition of hidden Markov processes in
statistics and signal processing is yet up to a degree different,
namely the observed process $(X_i)_{i\in\mathbb{Z}}$ depends on the
hidden process $(Y_i)_{i\in\mathbb{Z}}$ via a probability distribution
and $X_i$ is conditionally independent of the other observables given
$Y_i$.  All the presented definitions are, however, equivalent and the
terminological discussion can be put aside.

In the following turn we inspect the mutual information. Having entropy
$H(X)=\sred{\kwad{-\log P(X)}}$ with $\log$ denoting the binary
logarithm throughout this paper, mutual information is defined as
$I(X;Y)=H(X)+H(Y)-H(X,Y)$. Here we will be interested in the block
mutual information of the observable process,
\begin{align}
  E(n):=I(X_{-n+1}^{0};X_1^n)
  ,
\end{align}
where $X_k^l$ denotes the block $(X_i)_{k\le i\le l}$. More
specifically, we are interested in processes for which excess entropy
$E=\lim_{n\rightarrow\infty} E(n)$ is infinite and $E(n)$ diverges at
a~power-law rate.  We want to show that such an effect is possible for
very simple hidden Markov processes. (Travers and Crutchfield
\cite{TraversCrutchfield11} considered some examples of nonergodic and
ergodic hidden Markov processes with infinite excess entropy but they
did not investigate the rate of divergence of $E(n)$.) Notice that by
the data processing inequality for the Markov process
$(Y_i)_{i\in\mathbb{Z}}$, we have
\begin{align}
  E(n)\le I(Y_{-n+1}^{0};Y_1^n)= I(Y_0;Y_1)\le H(Y_0)
  .
\end{align}
Thus the block mutual information $E(n)$ may diverge only if the
entropy of the hidden state is infinite. To achieve this effect, the
hidden variable $Y_0$ must necessarily assume an infinite number of
values.

Now we introduce our class of examples.  Let us assume that hidden
states $\sigma_{nk}$ may be grouped into levels
\begin{align}
  T_n:=\klam{\sigma_{nk}}_{1\le k\le r(n)}
\end{align}
that comprise equiprobable values. Moreover, we suppose that the level
indicator
\begin{align}
  N_i:=n \iff Y_i\in T_n
\end{align}
has distribution 
\begin{align}
  \label{MarginN}
  P(N_i=n)=\frac{C}{n\log^\alpha n} 
  .
\end{align}
For $\alpha\in(1,2]$, entropy $H(N_i)$ is infinite and so is
$H(Y_i)\ge H(N_i)$ since $N_i$ is a function of $Y_i$. In the
following, we work with this specific distribution of $Y_i$.

As we will show, the rate of growth of the block mutual information
$E(n)$ is bounded in terms of exponent $\alpha$ from equation
(\ref{MarginN}). Let us write $f(n)=O(g(n))$ if $f(n)\le K g(n)$ for
a~$K>0$ and $f(n)=\Theta(g(n))$ if $K_1 g(n)\le f(n)\le K_2 g(n)$ for
$K_1,K_2>0$.
\begin{theorem}
  \label{theoEntropyExcess}
  Assume that $\mathbb{Y}=\klam{\sigma_{nk}}_{1\le k\le r(n), n\ge
    2}$, where function $r(n)$ satisfies $r(n)=O(n^p)$ for
  a~$p\in\mathbb{N}$. Moreover assume that
  \begin{align}
    \label{MarginY}
    P(Y_i=\sigma_{nk})=\frac{1}{r(n)}\cdot\frac{C}{n\log^\alpha n}
    ,
  \end{align}
  where $\alpha\in(1,2]$ and $C^{-1}=\sum_{n=2}^\infty (n\log^\alpha
  n)^{-1}$.  Then we have
  \begin{align}
    \label{EntropyExcess}
    E(n)=
    \begin{cases}
      O\okra{n^{2-\alpha}}, & \alpha\in(1,2),
      \\
      O\okra{\log n}, & \alpha=2.
    \end{cases}
  \end{align}
\end{theorem}

The interesting question becomes whether there exist hidden Markov
processes that achieve the upper bound established in Theorem
\ref{theoEntropyExcess}. If so, can they be ergodic? The answer to
both questions is positive and we will exhibit some simple examples of
such processes.

The first example that we present is nonergodic and the mutual
information diverges slower than expected from Theorem
\ref{theoEntropyExcess}.
\begin{example}[Heavy Tailed Periodic Mixture I]
  \label{exOne}
  This example has been introduced in \cite{TraversCrutchfield11}.  We
  assume $\mathbb{Y}=\klam{\sigma_{nk}}_{1\le k\le r(n), n\ge 2}$,
  where $r(n)=n$.  Then we set the transition probabilities
  \begin{align}
    P(Y_{i+1}=\sigma_{nk}|Y_i=\sigma_{ml})
    =
    \begin{cases}
      \boole{n=m,k=l+1}, & 1\le l \le m-1
      ,
      \\
      \boole{n=m,k=1}, & l=m
      .
    \end{cases}
  \end{align}
  We can see that the transition graph associated with the process
  $(Y_i)_{i\in\mathbb{Z}}$ consists of disjoint cycles on levels
  $T_n$. The stationary distribution of the Markov process is not
  unique and the process is nonergodic if more than one cycle has
  a~positive probability. Here we assume the cycle distribution
  (\ref{MarginN}) so the stationary marginal distribution of $Y_i$
  equals (\ref{MarginY}).  Moreover, the observable process is set as
  \begin{align}
    X_i
    =
    \begin{cases}
      0, & Y_i=\sigma_{nk},\, 1\le k \le n-1
      ,
      \\
      1, & Y_i=\sigma_{nn}
      .
    \end{cases}
  \end{align}
\end{example}

In the above example, the level indicator $N_i$ has infinite entropy
and is measurable with respect to the shift invariant algebra of the
observable process $(X_i)_{i\in\mathbb{Z}}$. Hence $E(n)$ tends to infinity
by the ergodic decomposition of excess entropy \cite[Theorem
5]{Debowski09}. A~more precise bound on the block mutual information
is given below.

\begin{proposition}
  \label{theoOne}
  For Example \ref{exOne}, we have 
  \begin{align}
    \label{EntropyExcessI}
    E(n)=
    \begin{cases}
      \Theta(\log^{2-\alpha} n), & \alpha\in(1,2),
      \\
      \Theta(\log \log n), & \alpha=2.
    \end{cases}
  \end{align}
\end{proposition}

The next example is also nonergodic but the rate of mutual information
reaches the upper bound. It seems to happen so because the information
about the hidden state level is coded in the observable process in a~more
concise way.

\begin{example}[Heavy Tailed Periodic Mixture II]
  \label{exTwo}
  We assume that $\mathbb{Y}=\klam{\sigma_{nk}}_{1\le k\le r(n), n\ge
    2}$, where $r(n)=s(n)$ is the length of the binary expansion of
  number $n$.  Then we set the transition probabilities
  \begin{align}
    P(Y_{i+1}=\sigma_{nk}|Y_i=\sigma_{ml})
    =
    \begin{cases}
      \boole{n=m,k=l+1}, & 1\le l \le s(m)-1
      ,
      \\
      \boole{n=m,k=1}, & l=s(m)
      .
    \end{cases}
  \end{align}
  Again, the transition graph associated with the process
  $(Y_i)_{i\in\mathbb{Z}}$ consists of disjoint cycles on levels
  $T_n$.  As previously, we assume the cycle distribution
  (\ref{MarginN}) and the marginal distribution
  (\ref{MarginY}). Moreover, let $b(n,k)$ be the $k$-th digit of the
  binary expansion of number $n$. (We have $b(n,1)=1$.)  The observable
  process is set as
  \begin{align}
    X_i
    =
    \begin{cases}
      2, & Y_i=\sigma_{n1}
      ,
      \\
      b(n,k), & Y_i=\sigma_{nk}, 2\le k \le s(n)
      .
    \end{cases}
  \end{align}
\end{example}

\begin{proposition}
  \label{theoTwo}
  For Example \ref{exTwo}, we have 
  \begin{align}
    \label{EntropyExcessII}
    E(n)=
    \begin{cases}
      \Theta(n^{2-\alpha}), & \alpha\in(1,2),
      \\
      \Theta(\log n), & \alpha=2.
    \end{cases}
  \end{align}
\end{proposition}

In the third example the rate of mutual information also reaches the
upper bound and the process is additionally ergodic. The process
resembles the Branching Copy (BC) process introduced in
\cite{TraversCrutchfield11}.  There are three main differences between
the BC process and our process. First, we discuss a~simpler
nonunifilar presentation of the process rather than a~more complicated
unifilar one. Second, we add strings of separators $(s(m)+1)\times 3$
in the observable process. Third, we put slightly different transition
probabilities to obtain a~simpler stationary distribution. All these
changes lead to a~simpler computation of mutual information.

\begin{example}[Heavy Tailed Mixing Copy]
  \label{exThree}
  Let $\mathbb{Y}=\klam{\sigma_{nk}}_{1\le k\le r(n), n\ge 2}$ with
  $r(n)=3s(n)$ and $s(n)$ being the length of the binary expansion
  of number $n$.  Then we set the transition probabilities
  \begin{align}
    P(Y_{i+1}=\sigma_{nk}|Y_i=\sigma_{ml})
    =
    \begin{cases}
      \boole{n=m,k=l+1}, & 1\le l \le r(m)-1
      ,
      \\
      p(n)\boole{k=1}, & l=r(m)
      ,
    \end{cases}
  \end{align}
  where 
  \begin{align}
    p(n)=\frac{1}{r(n)}\cdot\frac{D}{n\log^\alpha n}
  \end{align}
  and $D^{-1}=\sum_{n=2}^\infty (r(n)\cdot n\log^\alpha n)^{-1}$. This
  time levels $T_n$ communicate through transitions
  $\sigma_{mr(m)}\rightarrow\sigma_{n1}$ happening with probabilities
  $p(n)$. The transition graph of the process $(Y_i)_{i\in\mathbb{Z}}$
  is strongly connected and there is a~unique stationary distribution.
  Hence the process is ergodic. It can be easily verified that the
  stationary distribution is (\ref{MarginY}) so the levels are
  distributed according to (\ref{MarginN}). As previously, let
  $b(n,k)$ be the $k$-th digit of the binary expansion of number
  $n$. The observable process is set as
  \begin{align}
    X_i
    =
    \begin{cases}
      2, & Y_i=\sigma_{n1}
      ,
      \\
      b(n,k), & Y_i=\sigma_{nk},\, 2\le k \le s(n)
      ,
      \\
      3, & Y_i=\sigma_{nk},\, s(n)+1\le k \le 2s(n)+1
      ,
      \\
      b(n,k-2s(n)), & Y_i=\sigma_{nk},\, 2s(n)+2\le k \le 3s(n)
      .
    \end{cases}
  \end{align}
\end{example}

\begin{proposition}
  \label{theoThree}
  For Example \ref{exThree}, $E(n)$ satisfies (\ref{EntropyExcessII}).
\end{proposition}

Resuming our results, we make this comment.  The power-law growth of
block mutual information has been previously considered a~hallmark of
stochastic processes that model ``complex behavior'', such as texts in
natural language
\cite{Hilberg90,BialekNemenmanTishby01b,CrutchfieldFeldman03}. However,
the constructed examples of hidden Markov processes feature quite
simple transition graphs. Consequently, one may doubt whether
power-law growth of mutual information is a~sufficient reason to call
a~given stochastic process a~model of complex behavior, even when we
restrict the class of processes to processes over a~finite
alphabet. Basing on our experience with other processes with rapidly
growing block mutual information
\cite{Debowski11b,Debowski11d,Debowski12,Debowski10}, which are more
motivated linguistically, we think that infinite excess entropy is
just one of the necessary conditions. Identifying other conditions for
stochastic models of complex systems is a~matter of further
interdisciplinary research. We believe that these conditions depend on
a particular system to be modeled.

\section{Proofs}
\label{secProofs}

We begin with two simple bounds. 
\begin{lemma}
  \label{theoSumAlpha}
  Let $\alpha\in(1,2]$. On the one hand, we have
  \begin{align}
    \label{SumaAlphaOne}
    \sum_{m=2}^{n}\frac{1}{m\log^{\alpha-1}m}
    =\delta_1+
    \begin{cases}
      \displaystyle
      \frac{\ln 2}{2-\alpha}\okra{\log^{2-\alpha} n-1}, & \alpha\in(1,2),
      \\
      (\ln^2 2)\log \log n, & \alpha=2,      
    \end{cases}
  \end{align}
  where $0\le \delta_1\le 1/2$. On the other hand, we have
  \begin{align}
    \label{SumaAlpha}
    \sum_{m=n}^{\infty}\frac{1}{m\log^{\alpha}m}
    =\delta_2+\frac{\ln 2}{\alpha-1} \log^{1-\alpha} n
  \end{align}
  where $0\le \delta_2\le (n\log^{\alpha}n)^{-1}$.
\end{lemma}
\begin{proof}
  For a~continuous decreasing function $f$ we have $\int_a^b f(m)dm\le
  \sum_{m=a}^b f(m)\le f(a)+\int_a^b f(m)dm$. Moreover,
  \begin{align*}
    \int_{2}^{n}\frac{dm}{m\log^{\alpha-1}m}
    &=
    \int_{1}^{\log n}\frac{(\ln 2) dp}{p^{\alpha-1}}
    =
    \begin{cases}
      \displaystyle
      \frac{\ln 2}{2-\alpha}\okra{\log^{2-\alpha}n-1}, & \alpha\in(1,2),
      \\
      (\ln^2 2)\log\log n, & \alpha=2, 
    \end{cases}
    \\
    \int_{n}^\infty\frac{dm}{m\log^{\alpha}m}
    &=
    \int_{\log n}^\infty\frac{(\ln 2) dp}{p^{\alpha}}
    =
    \frac{\ln 2}{\alpha-1} \log^{1-\alpha}n
    .
  \end{align*}
  Hence the claims follow.
\end{proof}

For an event $B$, let us introduce conditional entropy $H(X|B)$ and
mutual information $I(X;Y|B)$, which are respectively the entropy of
variable $X$ and mutual information between variables $X$ and $Y$
taken with respect to probability measure $P(\cdot|B)$. The
conditional entropy $H(X|Z)$ and information $I(X;Y|Z)$ for a~variable
$Z$ are the averages of expressions $H(X|Z=z)$ and $I(X;Y|Z=z)$ taken
with weights $P(Z=z)$.  That is the received knowledge. Now comes
a~handy fact that we will also use.  Let $I_B$ be the indicator
function of event $B$. Observe that
\begin{align}
  I(X;Y)
  &=
  I(X;Y|I_B)+I(X;Y;I_B)
  \nonumber\\  
  &= 
  P(B)I(X;Y|B)+P(B^c)I(X;Y|B^c)+I(X;Y;I_B)
  \label{PBIB}
  ,
\end{align}
where the triple information $I(X;Y;I_B)$ satisfies
$\abs{I(X;Y;I_B)}\le H(I_B)\le 1$ by the information diagram
\cite{Yeung02}.

\begin{proof*}{Theorem \ref{theoEntropyExcess}}
  Consider the event $B=(N_o\le 2^n)$, where $N_0$ is the level
  indicator of variable $Y_0$.  On the one hand, by Markovianity of
  $(Y_i)_{i\in\mathbb{Z}}$, we have
  \begin{align*}
    I(X_{-n+1}^0;X_1^n|B)
    &\le
    I(Y_{-n+1}^0;Y_1^n|B)
    \\
    &\le
    I(Y_0;Y_1|B)
    \le
    H(Y_0|B)
    .
  \end{align*}
  On the other hand, for $B^c$, the complement of $B$, we have
  \begin{align*}
    I(X_{-n+1}^0;X_1^n|B^c)
    &\le
    H(X_{-n+1}^0|B^c)
    \le
    n\log\abs{\mathbb{X}}
    ,
  \end{align*}
  where $\abs{\mathbb{X}}$, the cardinality of set $\mathbb{X}$, is finite.
  Hence, using (\ref{PBIB}), we obtain
  \begin{align}
    E(n)
    &\le
    P(B)I(X_{-n+1}^0;X_1^n|B)
    +
    P(B^c)I(X_{-n+1}^0;X_1^n|B^c)
    + 1
    \nonumber\\
    &\le
    P(B) H(Y_0|B)+ nP(B^c)\log\abs{\mathbb{X}}+1
    \label{EnBound}
    ,
  \end{align}
  where
  \begin{align*}
    P(B)=\sum_{m=2}^{2^n}\frac{C}{m\log^{\alpha}m}
    .
  \end{align*}
  Using (\ref{SumaAlphaOne}) yields further
  \begin{align*}
    P(B) H(Y_0|B)
    &=
    P(B)
    \sum_{m=2}^{2^n}\frac{C}{P(B)\cdot m\log^{\alpha}m}
    \log \frac{P(B)\cdot r(m)\cdot m \log^{\alpha}m}{C}
    \\
    &=
    \sum_{m=2}^{2^n}\frac{C}{m\log^{\alpha}m}
    \log \frac{r(m)\cdot m \log^{\alpha}m}{C}
    + P(B) \log P(B)
    \\
    &=
    \Theta\okra{\sum_{m=2}^{2^n}\frac{1}{m\log^{\alpha-1}m}}
    =
    \begin{cases}
      \Theta\okra{n^{2-\alpha}}, & \alpha\in(1,2),
      \\
      \Theta\okra{\log n}, & \alpha=2.
    \end{cases}
  \end{align*}
  On the other hand, by (\ref{SumaAlpha}), we have
  \begin{align*}
    nP(B^c)
    &=
    n\sum_{m=2^n+1}^\infty\frac{C}{m\log^{\alpha}m}
    =
    \Theta\okra{n^{2-\alpha}}
    .
  \end{align*}
  Plugging both bounds into (\ref{EnBound}) yields the requested bound
  (\ref{EntropyExcess}).
\end{proof*}

Now we prove Propositions \ref{theoOne}--\ref{theoThree}. The proofs
are very similar and consist in constructing variables $D_n$ that are
both functions of $X_{-n+1}^0$ and functions of $X_1^n$. Given this
property, we obtain
\begin{align}
  E(n)
  &=
  I(X_{-n+1}^0,D_n;X_1^n)
  =
  I(D_n;X_1^n)+I(X_{-n+1}^0;X_1^n|D_n)
  \nonumber
  \\
  \label{EnDn}
  &=
  H(D_n)+I(X_{-n+1}^0;X_1^n|D_n)
  .
\end{align}
Hence, some lower bounds for the block mutual information $E(n)$
follow from the respective bounds for the entropies of $D_n$.

\begin{proof*}{Proposition \ref{theoOne}}
  Introduce random variable
  \begin{align*}
    D_n
    =
    \begin{cases}
      N_0, & 2N_0\le n,
      \\
      0, & 2N_0> n.
    \end{cases}
  \end{align*}
  Equivalently, we have
  \begin{align*}
    D_n
    =
    \begin{cases}
      N_1, & 2N_1\le n,
      \\
      0, & 2N_1> n.
    \end{cases}
  \end{align*}

  It can be seen that $D_n$ is both a~function of $X_{-n+1}^0$ and
  a~function of $X_1^n$.  On the one hand, observe that if $2N_0\le n$
  then we can identify $N_0$ given $X_{-n+1}^0$ because the full
  period is visible in $X_{-n+1}^0$, bounded by two delimiters $1$. On
  the other hand, if $2N_0>n$ then given $X_{-n+1}^0$ we may conclude
  that the period's length $N_0$ exceeds $n/2$, regardless whether the
  whole period is visible or not. Hence variable $D_n$ is a~function
  of $X_{-n+1}^0$.  In a~similar way, we show that $D_n$ is a~function
  of $X_1^n$. Given both facts, we derive (\ref{EnDn}).

  Next, we bound the terms appearing on the right hand side of
  (\ref{EnDn}). For a~given $N_0$, variable $X_{-n+1}^0$ assumes at
  most $N_0$ distinct values, which depend on $N_0$. 
  Hence 
  \begin{align*}
    H(X_{-n+1}^0|D_n=m)\le \log m \text{ for $2\le m\le\floor{n/2}$}
    .
  \end{align*}
  On the other hand, if we know that $N_0> n$ then the number of
  distinct values of variable $X_{-n+1}^0$ equals $n+1$. Consequently,
  if we know that $D_n=0$, i.e., $N_0\ge \floor{n/2}+1$, then the
  number of distinct values of $X_{-n+1}^0$ is bounded above by
  \begin{align*}
    n+1+\sum_{m=\floor{n/2}+1}^{n} m
    &=
    n+1+\frac{n(n+1)}{2}+\frac{\floor{n/2}(\floor{n/2}+1)}{2}
    \\
    &\le 
    \frac{3n^2+14n+8}{8}\le \frac{25}{8}n^2
    .
  \end{align*}
  In this way we obtain
  \begin{align*}
    H(X_{-n+1}^0|D_n=0)\le \log (25n^2/8)
    .
  \end{align*}

  Hence, by (\ref{SumaAlphaOne}) and (\ref{SumaAlpha}), the
  conditional mutual information may be bounded
  \begin{align*}
    I(X_{-n+1}^0;X_1^n|D_n) &\le H(X_{-n+1}^0|D_n)
    \\
    &=
    \sum_{m=2}^{\floor{n/2}} P(D_n=m) H(X_{-n+1}^0|D_n=m)    
    \\
    &\phantom{==} + 
    P(D_n=0) H(X_{-n+1}^0|D_n=0)
    \\
    &\le 
    \sum_{m=2}^{\floor{n/2}}\frac{C\log m}{m\log^{\alpha}m}
    +
    \sum_{m=\floor{n/2}+1}^{\infty}\frac{C\log(25n^2/8)}{m\log^{\alpha}m}
    \\
    &=
    \begin{cases}
      \Theta(\log^{2-\alpha} n), & \alpha\in(1,2),
      \\
      \Theta(\log \log n), & \alpha=2.
    \end{cases}    
  \end{align*}
  The entropy of $D_n$ may be bounded similarly,
  \begin{align*}
    H(D_n)
    &=\sum_{m=2}^{\floor{n/2}}\frac{C}{m\log^{\alpha}m}
    \log \frac{m \log^{\alpha}m}{C} - P(D_n=0)\log P(D_n=0)
    \\
    &=
    \begin{cases}
      \Theta(\log^{2-\alpha} n), & \alpha\in(1,2),
      \\
      \Theta(\log \log n), & \alpha=2.
    \end{cases}        
  \end{align*}
  Hence, because $E(n)$ satisfies (\ref{EnDn}), we obtain
  (\ref{EntropyExcessI}).
\end{proof*}

\begin{proof*}{Proposition \ref{theoTwo}}
  Introduce random variable
  \begin{align*}
    D_n
    =
    \begin{cases}
      N_0, & 2s(N_0)\le n,
      \\
      0, & 2s(N_0)> n.
    \end{cases}
  \end{align*}
  Equivalently, we have
  \begin{align*}
    D_n
    =
    \begin{cases}
      N_1, & 2s(N_1)\le n,
      \\
      0, & 2s(N_1)> n.
    \end{cases}
  \end{align*}

  As in the previous proof, the newly constructed variable $D_n$ is
  both a~function of $X_{-n+1}^0$ and a~function of $X_1^n$.  If
  $2s(N_0)\le n$ then we can identify $N_0$ given $X_{-n+1}^0$ because
  the full period is visible in $X_{-n+1}^0$, bounded by two
  delimiters $2$. If $2s(N_0)>n$ then given $X_{-n+1}^0$ we may
  conclude that the period's length $s(N_0)$ exceeds $n/2$, regardless
  whether the whole period is visible or not. Hence variable $D_n$ is
  a~function of $X_{-n+1}^0$.  In a~similar way, we demonstrate that
  $D_n$ is a~function of $X_1^n$. By these two facts, we infer
  (\ref{EnDn}).

  Observe that the largest $m$ such that $s(m)=\floor{\log m}+1\le
  \floor{n/2}$ is $m=2^{\floor{n/2}}-1$.  Using (\ref{SumaAlphaOne}), the
  entropy of $D_n$ may be bounded as
  \begin{align*}
    H(D_n)
    &=\sum_{m=2}^{2^{\floor{n/2}}-1}\frac{C}{m\log^{\alpha}m}
    \log \frac{m \log^{\alpha}m}{C} - P(D_n=0)\log P(D_n=0)
    \\
    &=
    \begin{cases}
      \Theta(n^{2-\alpha}), & \alpha\in(1,2),
      \\
      \Theta(\log n), & \alpha=2,
    \end{cases}        
  \end{align*}
  Thus (\ref{EntropyExcessII}) follows by (\ref{EnDn}) and Theorem
  \ref{theoEntropyExcess}.
\end{proof*}

\begin{proof*}{Proposition \ref{theoThree}}
    Introduce random variable
  \begin{align*}
    D_n
    =
    \begin{cases}
      m, & Y_0=\sigma_{mk},\, s(m)+1\le k \le 2s(m),\, 2s(m)\le n,
      \\
      0, & \text{else}.
    \end{cases}
  \end{align*}
  Equivalently, we have
  \begin{align*}
    D_n
    =
    \begin{cases}
      m, & Y_1=\sigma_{mk},\, s(m)+2\le k \le 2s(m)+1,\, 2s(m)\le n,
      \\
      0, & \text{else}.
    \end{cases}
  \end{align*}
  
  Again, it can be seen that $D_n$ is both a~function of $X_{-n+1}^0$ and
  a~function of $X_1^n$.  The way of computing $D_n$ given
  $X_{-n+1}^0$ is as follows.  If
  \begin{align*}
    X_{-n+1}^0=(...,2,b(m,2),b(m,3),...,b(m,s(m)),
    \underbrace{3,...,3}_{\text{$l$ times}})
  \end{align*}
  for some $m$ such that $2s(m)\le n$ and $1\le l\le s(m)$ then we return
  $D_n=m$. Otherwise we return $D_n=0$. The recipe for $D_n$ given
  $X_1^n$ is mirror-like. If
  \begin{align*}
    X_1^n=(\underbrace{3,...,3}_{\text{$l$ times}}, 
    b(m,2),b(m,3),...,b(m,s(m)),2,...)
  \end{align*}
  for some $m$ such that $2s(m)\le n$ and $1\le l\le s(m)$ then we return
  $D_n=m$. Otherwise we return $D_n=0$. In view of these observations
  we derive (\ref{EnDn}), as in the previous two proofs. 

  Now, for $m\neq 0$ and $s(m)\le n/2$, the distribution of $D_n$ is
  \begin{align*}
    P(D_n=m)
    =
    \frac{s(m)}{3s(m)}\cdot\frac{C}{m\log^\alpha m}
    =
    \frac{1}{3}\cdot\frac{C}{m\log^\alpha m}
    .
  \end{align*}
  Notice that the largest $m$ such that $s(m)=\floor{\log m}+1\le
  \floor{n/2}$ is $m=2^{\floor{n/2}}-1$.  Hence, by (\ref{SumaAlphaOne}), the
  bound for the entropy of $D_n$ is
  \begin{align*}
    H(D_n)
    &=\sum_{m=2}^{2^{\floor{n/2}}-1}\frac{C}{3 m\log^{\alpha}m}
    \log \frac{3 m \log^{\alpha}m}{C} - P(D_n=0)\log P(D_n=0)
    \\
    &=
    \begin{cases}
      \Theta(n^{2-\alpha}), & \alpha\in(1,2),
      \\
      \Theta(\log n), & \alpha=2,
    \end{cases}        
  \end{align*}
  Consequently, (\ref{EntropyExcessII}) follows by (\ref{EnDn}) and
  Theorem \ref{theoEntropyExcess}.
\end{proof*}

\section*{Acknowledgment}

I~thank Nick Travers, Jan Mielniczuk, and an anonymous referee for
comments and remarks.

\bibliographystyle{abbrvnat}


\end{document}